\documentclass[10pt,letterpaper]{article}
\usepackage[T1]{fontenc}
\usepackage{lmodern}
\usepackage[margin=0.83in,top=0.78in,bottom=0.78in]{geometry}
\usepackage{amsmath,amssymb,amsthm,graphicx,microtype}
\usepackage{enumitem}
\usepackage{xcolor}
\usepackage[colorlinks=true,linkcolor=blue!50!black,citecolor=blue!50!black,urlcolor=blue!50!black]{hyperref}
\usepackage[font=small,labelfont=bf]{caption}
\usepackage{titlesec}
\titleformat{\section}{\normalfont\large\bfseries}{\thesection.}{0.55em}{}
\titlespacing*{\section}{0pt}{1.1ex plus 0.5ex minus 0.2ex}{0.6ex plus 0.2ex}
\newcommand{\ket}[1]{\lvert #1\rangle}

\newcommand{\Tr}{\operatorname{Tr}}
\newcommand{\ii}{\mathrm{i}}
\newcommand{\ind}{\mathbf{1}}
\newtheorem{proposition}{Proposition}
\hypersetup{pdftitle={The coupler Eve did not monitor: non-signalling reports from interacting cavities},pdfauthor={Yueying Wu}}
\title{\vspace{-1.4em}\LARGE\bfseries The coupler Eve did not monitor:\\[0.2em]
\Large non-signalling reports from interacting cavities}
\author{Yueying Wu\\[0.2em]
\small Department of Physics, University of Illinois Urbana-Champaign, USA}
\date{\small October 22, 2026}

\begin{document}
\maketitle
\vspace{-1.3em}
\begin{abstract}
Eve receives genuine average photon counts from two opaque devices, finds input-independent single-party report distributions, and sees little activity in a specified leakage monitor. Do these observations establish physical isolation? We construct a cavity-based counterexample. Independently operated controls activate a shared interaction after the inputs arrive. Internal block measurements yield a CHSH score $S=13/4$ for four repetitions, approaching $4$ for large blocks. A fresh shared random bit symmetrizes the binary results without changing their parity; local preparation of vacuum or one-photon output states then produces truthful photon-number reports with exactly non-signalling marginals. A dispersive-bus model also permits weak monitored leakage; an illustrative dissipative calculation gives $S\simeq3.1702$ with ideal output encoding. These scores exceed $2\sqrt2$ because the devices interact during their response, not because quantum mechanics is violated. The construction separates honest measurement, non-signalling reported statistics, and causal isolation, while making the additional randomness, preparation, and monitoring assumptions explicit.
\end{abstract}

\section{Eve's view: two black boxes and a quiet monitor}

Eve supplies Alice and Bob with independent, uniformly random bits $x$ and $y$. Each party has an opaque cavity device, a control knob, and a photon counter. They return only average photon counts from designated local output modes. Eve also watches a specified leakage channel and tests whether either report distribution depends on the other party's input. Our question is what these checks establish about the causal structure inside the boxes.

Our proposed explanation is mundane at the level of causality but concrete at the level of hardware: a shared coupling element was never included in Eve's audit. Her two black boxes have two input ports and two report ports, but their physical interiors are not independent. For example, both modules may connect to an unaudited bus in a common enclosure. The restricted access, monitoring channel, and missed connection are assumptions of this scenario; they are not claims that a real apparatus can defeat arbitrary inspection.

One game trial is an entire \emph{block}. Eve chooses $(x,y)$ once. The devices perform $M$ independently reset internal repetitions with that pair, producing
\begin{equation}
 \bar n_A=\frac1M\sum_{j=1}^M n_{A,j},\qquad
 \bar n_B=\frac1M\sum_{j=1}^M n_{B,j}.
 \label{eq:reports}
\end{equation}
These internal averages are \emph{not} the final reports. Each device thresholds its own average, applies a common random output flip, and physically prepares and measures a local output state. Eve receives the resulting output-mode sample averages $m_A,m_B$, defined in Sec.~\ref{sec:encoding}. This added preparation stage changes the protocol: it does not relabel the original cavity counts as different measured values.

The internal state preparation and phase reference are input independent; ideal copies can be supplied before the inputs. A fresh shared fair bit per block is likewise supplied in advance, independently of the inputs and internal records. It is not disclosed to Eve. The internal block length, output readout length, threshold rule, and deadline are fixed in advance, and every block is retained. Both the internal measurements and the final reports are genuine; no expectation value is inferred exactly from an arbitrary single-photon sample.

The intended rule prohibits communication from input delivery until the reports are fixed. The cheat deliberately violates that rule through the shared interaction. Calling devices black boxes removes the need to know their internal implementation; it does not establish the causal separation required for a Bell test~\cite{brunner2014}.

\section{Inside the boxes: two knobs, one interaction}

Let $a,b$ denote the cavity modes. In an ideal rotating-frame description, the two controls activate
\begin{equation}
 H_{xy}(t)=\hbar xy\,g(t)(a^\dagger b+b^\dagger a),
 \qquad \Theta=\int_0^T g(t)\,dt=\frac\pi4.
 \label{eq:H}
\end{equation}
Alice sets $x$ and Bob sets $y$. Neither has to learn the other's input before operating a knob: the hardware implements the product $xy$. Independent controls do not imply independent dynamics.

Each repetition starts in
\begin{equation}
 \ket{\psi_0}=\frac{\ket{1,0}+\ii\ket{0,1}}{\sqrt2}.
 \label{eq:initial}
\end{equation}
For instance, a photon in $A$ followed by a preparation pulse with mixing angle $-\pi/4$ produces this state. With $\theta=xy\Theta$, the response pulse gives
\begin{equation}
 \ket{\psi_{xy}}=
 \frac{\cos\theta+\sin\theta}{\sqrt2}\ket{1,0}
 +\ii\frac{\cos\theta-\sin\theta}{\sqrt2}\ket{0,1},
 \qquad
 \mu_{A,B}^{xy}=\frac{1\pm\sin(2xy\Theta)}2.
 \label{eq:response}
\end{equation}
For $00$, $01$, and $10$, both conditional means are $1/2$. For $11$, the state becomes $\ket{1,0}$ and the means are $(1,0)$. The response pulse is stopped at the prescribed area; leaving the exchange on indefinitely would produce oscillations. The initial coherence $\langle a^\dagger b\rangle=\ii/2$ sets the direction of transfer. An equal incoherent mixture would remain balanced, and two initially empty cavities would remain empty.

A demonstrated hardware ingredient is the two-pump cavity beam splitter of Gao \emph{et al.}~\cite{gao2018}. Through a shared nonlinear element, pump amplitudes $x\xi_1$ and $y\xi_2$ can generate exchange proportional to $xy\xi_1\xi_2^*$. Assigning one control to each player supplies the proposed two-knob interface. The control fields still reach common hardware. Single-pump frequency shifts, unwanted exchange, and losses must be included in a device-specific implementation. That experiment establishes a coupling mechanism, not its concealment from Eve.

\section{The internal block produces the score}

Each device applies a fixed threshold to its own internal average:
\begin{equation}
 u=\ind[\bar n_A>3/4],\qquad v=\ind[\bar n_B>3/4],
 \qquad \text{target parity }u\oplus v=xy.
 \label{eq:threshold}
\end{equation}
These internal bits define the block correlators below. Section~\ref{sec:encoding} gives final photon-number reports from which Eve obtains binary outputs with exactly the same parity. Thus the score is an ordinary binary-output CHSH expression at the \emph{block} level~\cite{chsh1969}:
\begin{equation}
 E_{xy}=\mathbb E[(-1)^{u+v}\mid x,y],\quad
 S_M=E_{00}+E_{01}+E_{10}-E_{11},\quad
 \omega_M=\frac12+\frac{S_M}{8}.
 \label{eq:chsh}
\end{equation}
For genuinely isolated devices, arbitrary local processing of a block is still a local measurement procedure. The standard bounds remain $S_M\leq2$ for local classical strategies and $S_M\leq2\sqrt2$ for quantum strategies with shared states and local measurements~\cite{tsirelson1980}. Averaging alone does not evade either bound.

In our interacting model, an inactive input has Alice's total count $K\sim\mathrm{Binomial}(M,1/2)$ and Bob's total count $M-K$. The thresholded bits are equal exactly when $M/4\leq K\leq3M/4$, in which case both are zero. Define
\begin{equation}
 q_M=2^{-M}\sum_{k=\lceil M/4\rceil}^{\lfloor3M/4\rfloor}\binom{M}{k}.
 \label{eq:q}
\end{equation}
Input $11$ gives $(u,v)=(1,0)$ with certainty. Consequently,
\begin{equation}
 \omega_M=\frac{3q_M+1}{4},\qquad S_M=6q_M-2.
 \label{eq:score}
\end{equation}
For $M=4$, the only inactive-input failures are four photons all counted at $A$ or all at $B$. Thus $q_4=7/8$ and
\begin{equation}
 \boxed{\omega_4=\frac{29}{32}=0.90625,\qquad S_4=\frac{13}{4}=3.25>2\sqrt2.}
 \label{eq:four}
\end{equation}
These are expected scores over repeated blocks, not a claim of statistical significance from four measurements. Hoeffding's inequality~\cite{hoeffding1963} gives $1-\omega_M\leq(3/2)e^{-M/8}$, so $S_M\to4$, although finite-block success remains below one. Integer thresholds make the approach nonmonotone for small $M$; see Fig.~\ref{fig:protocol}(a).

This numerical exceedance of the quantum benchmark is real for the stated interacting model. Its interpretation as stronger-than-quantum Bell nonlocality would be false: Eq.~\eqref{eq:H} acts jointly after the settings arrive. No fair-sampling or selected-block effect is needed. By contrast, applying parity directly to each raw single-photon count would give $E_{xy}=-1$ for all four inputs and $|S|=2$. The block threshold is an essential part of the score being reported.

\begin{figure}[t]
 \centering
 \includegraphics[width=\linewidth]{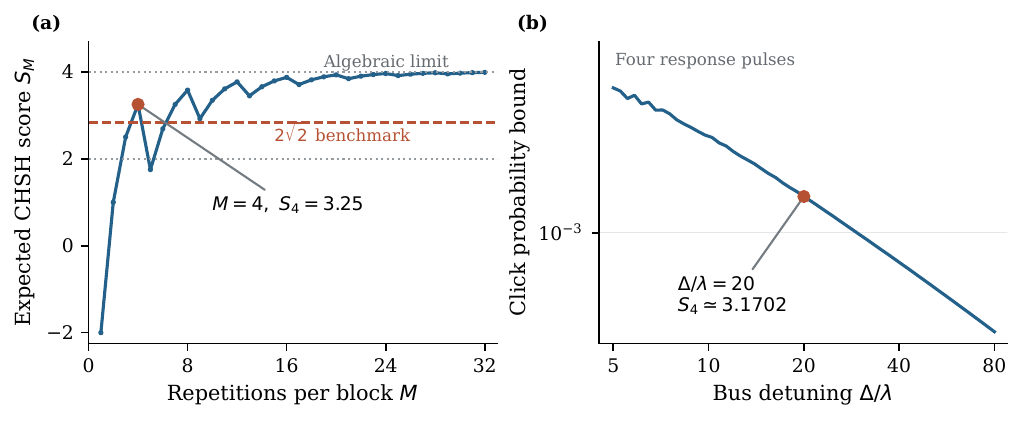}
 \caption{Theoretical predictions, with all outcomes retained. (a) Expected block CHSH score, preserved by shared random flips and ideal physical output encoding. The marked point is $M=4$, $S_4=3.25$; the dashed line is the isolated-device quantum benchmark. Final reports have uniform marginals for every $M$. (b) Upper bound $4\max_{xy}\mathbb E[N_{\rm det}\mid x,y]$ on at least one monitored click during four response pulses, from the dissipative bus model in Sec.~\ref{sec:monitor}. Parameters are $\kappa_C/\lambda=10^{-2}$, $\kappa_s/\lambda=10^{-4}$, and $\eta=1/2$. At $\Delta/\lambda=20$, $S_4\simeq3.1702$ before output errors and the click bound is $1.482\times10^{-3}$. This bound excludes preparation, reset, and output-readout stages.}
 \label{fig:protocol}
\end{figure}

\section{Honest, non-signalling output reports}
\label{sec:encoding}

Without an extra step, the internal averages signal: for example, $\mathbb E[\bar n_A\mid1,0]=1/2$ but $\mathbb E[\bar n_A\mid1,1]=1$. Alice and Bob instead use their fresh shared bit $r\in\{0,1\}$ to form
\begin{equation}
 \alpha=u\oplus r,\qquad \beta=v\oplus r,\qquad
 \alpha\oplus\beta=u\oplus v.
 \label{eq:mask}
\end{equation}
Simultaneous random output flips are a standard way to make binary marginals uniform while preserving correlators; see, for example, Appendix C of Ref.~\cite{bowles2015}. Their role here is to remove signalling from the final report interface, including when the internal distribution itself signals.

\begin{proposition}[Symmetrization of binary reports]
For any conditional internal distribution $P(u,v\mid x,y)$, a shared fair bit independent of $(x,y,u,v)$ produces
\begin{equation}
 Q(\alpha,\beta\mid x,y)=\tfrac12 P(\alpha,\beta\mid x,y)
 +\tfrac12 P(1-\alpha,1-\beta\mid x,y).
 \label{eq:twirl}
\end{equation}
Both marginals of $Q$ are uniform for every input pair, and all four parity correlators are unchanged.
\end{proposition}
\begin{proof}
Summing Eq.~\eqref{eq:twirl} over $\beta$ gives
$[P_A(\alpha\mid x,y)+P_A(1-\alpha\mid x,y)]/2=1/2$;
the same holds for Bob. Equation~\eqref{eq:mask} preserves parity in each realization, not just on average.
\end{proof}

To make these outputs actual photon-number reports, each party locally prepares an output mode in $\ket{\alpha}$ or $\ket{\beta}$, respectively. These can be separate modes or reset and repurposed working cavities after the interaction has stopped. For a prescribed $L\geq1$, repeat the local preparation and photon-number measurement $L$ times with the same encoded bit, and report
\begin{equation}
 m_A=\frac1L\sum_{\ell=1}^L n^{\rm out}_{A,\ell}=\alpha,
 \qquad
 m_B=\frac1L\sum_{\ell=1}^L n^{\rm out}_{B,\ell}=\beta
 \label{eq:outputmeans}
\end{equation}
in the ideal preparation/readout limit. Number eigenstates make these sample averages deterministic conditional on the encoded bits. Eve thresholds each report at $3/4$ and obtains $\alpha,\beta$. The reports are classical numbers; the output cavities need not be sent to Eve. The physical preparation, energy supply, and readout are extra local resources, not free consequences of the passive exchange Hamiltonian.

Consequently, the entire single-party report laws, not only their first moments, obey
\begin{equation}
 \Pr(m_i=0\mid x,y)=\Pr(m_i=1\mid x,y)=\tfrac12,
 \qquad \mathbb E[m_i\mid x,y]=\tfrac12\quad(i=A,B).
 \label{eq:uniform}
\end{equation}
The equality of these across-block means does not erase the correlations between the two reports within each block. For $xy=0$, Eq.~\eqref{eq:q} gives
\begin{equation}
 Q(0,0)=Q(1,1)=q_M/2,\qquad
 Q(0,1)=Q(1,0)=(1-q_M)/2;
 \label{eq:finaloff}
\end{equation}
for $xy=1$, only $(0,1)$ and $(1,0)$ occur, each with probability $1/2$. At $M=4$ the inactive-input probabilities are $7/16,7/16,1/16,1/16$, respectively. The score remains $13/4$. As $M\to\infty$, the full ideal report distribution approaches the Popescu--Rohrlich (PR) box~\cite{popescu1994}, not merely its CHSH value. This is a simulation using an actual interaction; no finite $M$ here attains unit winning probability.

Fresh randomness and a restricted interface matter. Conditioning on a disclosed $r$ can restore the original signalling. Revealing the internal counts also defeats the construction. The proposition concerns the final reports averaged over $r$, not all environmental records. Fixed local preparation/readout channels depending only on the encoded bit preserve no-signalling, although they can reduce the score. Correlated errors or leakage dependent on unreported internal records require a separate analysis.

\section{Why the monitored port can remain quiet}
\label{sec:monitor}

To model a possible monitoring blind spot, introduce a third mode $c$ with positive detuning $\Delta$ and independently switched couplings of magnitude $\lambda$:
\begin{equation}
 \frac{H^{(3)}_{xy}}\hbar
 =\Delta c^\dagger c+\lambda x(a^\dagger c+c^\dagger a)
 -\lambda y(b^\dagger c+c^\dagger b).
 \label{eq:bus}
\end{equation}
The relative minus sign is a fixed coupling phase. This is a separate minimal model for the monitoring argument, not a microscopic reconstruction of Ref.~\cite{gao2018}. Dispersive, virtually populated buses are established interaction mechanisms~\cite{majer2007}.

For $\Delta\gg\lambda$, elimination of $c$ to second order gives
\begin{equation}
 \frac{H_{\rm eff}}\hbar
 =-\frac{\lambda^2}{\Delta}(x n_A+y n_B)
 +\frac{\lambda^2}{\Delta}xy(a^\dagger b+b^\dagger a).
 \label{eq:effective}
\end{equation}
For $11$, the number shifts are equal and contribute only a common phase in the one-excitation sector. For an inactive input they change phases but not populations in this effective model. Hence $g_{\rm eff}=\lambda^2/\Delta$ supplies the desired response at
\begin{equation}
 T=\frac{\pi\Delta}{4\lambda^2},\qquad
 \langle n_C(t)\rangle=O(\lambda^2/\Delta^2).
 \label{eq:scaling}
\end{equation}
The exact lossless occupation is given in the Appendix. At finite detuning there are small bus excursions and unwanted population changes even for single-control inputs: exact mean equality belongs to the ideal effective protocol, not to the full three-mode dynamics.

Suppose Eve's monitor counts photons lost from the bus through one port with rate $\kappa_{\rm mon}$, collecting a fraction $\eta$. With $\rho(t)$ the unconditional state including loss,
\begin{equation}
 \mathbb E[N_{\rm det}]
 =\eta\kappa_{\rm mon}\int_0^T\Tr[n_C\rho(t)]\,dt,
 \qquad
 \Pr(N_{\rm det}\geq1)\leq\mathbb E[N_{\rm det}].
 \label{eq:monitor}
\end{equation}
In the dispersive, weak-error regime, Eqs.~\eqref{eq:scaling}--\eqref{eq:monitor} imply
\begin{equation}
 \mathbb E[N_{\rm det}]=O(\eta\kappa_{\rm mon}/\Delta),
 \qquad g_{\rm eff}T=\pi/4.
 \label{eq:quiet}
\end{equation}
A substantial coherent action and a small expected leakage count can therefore coexist. The claim concerns a specified detector, not the absence of electromagnetic fields or of a physical connection.

For a finite-parameter check, evolve Eq.~\eqref{eq:bus} with Lindblad jump operators $\sqrt{\kappa_s}a$, $\sqrt{\kappa_s}b$, and $\sqrt{\kappa_C}c$, using $\Delta=20\lambda$, $\kappa_C=10^{-2}\lambda$, $\kappa_s=10^{-4}\lambda$, and $\eta=1/2$. Take $\kappa_{\rm mon}=\kappa_C$ and the duration in Eq.~\eqref{eq:scaling}, resetting to $\ket{\psi_0}\ket0_C$ each time. Direct evolution gives $S_4\simeq3.1702$ and a click-probability bound below $1.49\times10^{-3}$ during the four internal response pulses. Symmetrization makes this noisy internal distribution non-signalling at the final interface and preserves its score with ideal output encoding. Cavity loss and residual bus occupation contribute zero internal cavity counts, not discarded trials. These parameters are illustrative, not a hardware calibration, and the calculation does not include output-stage errors or radiation.

Repeated blocks accumulate detection opportunities. Across $R$ response pulses, the union bound is $O(R\eta\kappa_{\rm mon}/\Delta)$, capped at one; quietness of a long experiment is not guaranteed. Larger detuning also lengthens $T$, so storage loss requires $\kappa_s\Delta/\lambda^2\ll1$. Control drives, state preparation, reset, readout radiation, other ports, and direct inspection are outside this monitor model. Their detectability must be assessed separately. In particular, blocks with monitor clicks are not removed from the CHSH calculation.

\section{What Eve can and cannot infer}

The final reports now pass an exact no-signalling check in the stated model, even though the working cavities interact. This removes the simplest marginal diagnostic; it does not make the interaction disappear. An adequately significant score above $2\sqrt2$ still rules out the model of isolated quantum devices with independent settings and the stated scoring procedure. Eve need not locate a coupler to reject that model, but the final marginals alone no longer locate the failure. A quiet specified port, honest output counts, and non-signalling report distributions do not jointly certify causal isolation.

Here a controlled randomization hides an available causal influence in the reported statistics. This is consistent with the role of fine-tuning in classical causal explanations of non-signalling Bell correlations~\cite{wood2015}; the biased-mask calculation in Appendix B makes that dependence explicit for this protocol.

The distinction between per-block reports and across-block means is essential. For any no-signalling reports with finite means, set $d_{xy}=|\mathbb E[m_A\mid x,y]-\mathbb E[m_B\mid x,y]|$. Then
\begin{equation}
 d_{11}\leq d_{10}+d_{00}+d_{01}.
 \label{eq:meanwitness}
\end{equation}
This follows by writing the local means as $A_x,B_y$ and applying the triangle inequality along $A_1,B_0,A_0,B_1$. Thus exact equality of across-block means for $00,01,10$ forces equality for $11$. The raw internal response violates this condition, while the final reports have $d_{xy}=0$ for every input. We have replaced the original mean-inequality target with a game about correlations between random, genuine photon-number reports.

The operational lesson is to audit causal access during the entire interval from input delivery to fixed outputs, rather than treating a passed marginal check as a certificate of separation. Suitable spacetime separation can close the communication loophole~\cite{hensen2015}. Delaying transmission of an already fixed report is harmless by itself; permitting interactions before that report is fixed is the issue. Monitoring extra channels or accessing internal records may reveal the implementation. No claim is made that the complete apparatus is unobservable.

The single-photon implementation uses a nonclassical state, but its internal mean pattern has a coherent-state counterpart. The input $\ket{\gamma}_A\ket{\ii\gamma}_B$ remains a product under the beam splitter~\cite{kim2002}, giving means $(N,N)$ or $(2N,0)$, with $N=|\gamma|^2>0$. Classical amplitudes obey the same transformation. Equation~\eqref{eq:score} specifically uses single-photon counting statistics. Shared output flips are classical, and communication can simulate these correlations. The proposed contribution is an explicit cavity realization and audit model, not a new symmetrization method, interaction mechanism, or quantum advantage.

A useful next experiment would compare raw and encoded marginals, measure the final block score and leakage record, and characterize output preparation and readout alongside the internal interaction. Comparisons under matched time and loss budgets would test whether the full implementation remains useful. Experimental and theoretical collaboration on such tests is welcome. Honest, non-signalling reports can come from devices that communicated.

\appendix
\section{Exact bus occupation and numerical checks}

For Eq.~\eqref{eq:bus} without loss and initial state $\ket{\psi_0}\ket0_C$, put $s=x+y$. For $s>0$, the bright cavity mode $(xa-yb)/\sqrt{s}$ couples to $c$ with strength $\sqrt{s}\lambda$ and initially carries probability $1/2$. The orthogonal mode is dark. Therefore, with $\Omega_s=\sqrt{\Delta^2+4s\lambda^2}$,
\begin{equation}
 n_C^{xy}(t)=\frac{2s\lambda^2}{\Omega_s^2}\sin^2\!\left(\frac{\Omega_s t}{2}\right),
 \qquad
 \int_0^T n_C^{xy}(t)\,dt
 =\frac{s\lambda^2}{\Omega_s^2}\left(T-\frac{\sin(\Omega_sT)}{\Omega_s}\right).
 \label{eq:exactbus}
\end{equation}
For $s=0$ the occupation is zero. In particular $n_C(t)\leq4\lambda^2/\Delta^2$ for this initial state. Substituting the response time yields the inverse-detuning exposure scaling. For nonzero loss, Eq.~\eqref{eq:monitor} instead uses the dissipative occupation; the lossless formula is not substituted as an exact dissipative result.

In the full one-excitation model, let $p_A,p_B$ be the probabilities of a photon at each cavity after one response. A repetition may also give neither cavity a photon. At $M=4$, thresholding gives $u=1$ only if all four counts at $A$ are one, and similarly for $B$. The events are disjoint, so $E_{xy}=1-2(p_A^4+p_B^4)$. The source checks the ideal unitary, four-shot counting, Eq.~\eqref{eq:exactbus}, and the dissipative example. It independently enumerates internal records and mask values to verify the final joint probabilities, uniform marginals, and unchanged scores. All curves are model predictions; no experimental data are reported.

\section{Output-interface imperfections}

Suppose the local encoding and readout stage has a fixed stochastic response $T_A(m\mid\alpha)$, and similarly for Bob. Independent local channels yield
\begin{equation}
 \Pr(m_A\mid x,y)=\tfrac12 T_A(m_A\mid0)+\tfrac12 T_A(m_A\mid1),
 \label{eq:localnoise}
\end{equation}
and the analogous Bob expression. Thus no-signalling of the final reports survives such local imperfections, including input-independent loss or readout noise, even when the reports cease to be exactly $0$ or $1$. Their means need not remain $1/2$, and the score must be recomputed. This statement assumes the output channel has no additional dependence on the internal record or remote input.

For one illustrative error model, let Eve's final thresholded bits suffer independent symmetric flips with probabilities $\epsilon_A,\epsilon_B$, independent of the inputs and encoded bits. Then
\begin{equation}
 S_{\rm obs}=(1-2\epsilon_A)(1-2\epsilon_B)S_M.
 \label{eq:outputerror}
\end{equation}
For equal errors $\epsilon<1/2$, the ideal $M=4$ score remains above $2\sqrt2$ when
$\epsilon<[1-\sqrt{2\sqrt2/(13/4)}]/2\simeq0.0336$.
This simple bit-error model is not a calibration of Fock-state preparation or photon loss.

An imperfect mask has a different effect. If $\Pr(r=1)=p$, independently of the internal data, then
\begin{equation}
 \Pr(\alpha=1\mid x,y)=p+(1-2p)\Pr(u=1\mid x,y).
 \label{eq:biasedmask}
\end{equation}
For a fixed local input, a remote-input dependence in the internal marginal is suppressed by $|1-2p|$, but need not vanish. Exact uniform marginals require the fair-mask assumption (or another exact symmetrization), and a fresh independent mask prevents reuse from exposing correlations across blocks.

\end{document}